\documentclass[a4paper,11pt]{article}

\usepackage[english]{babel}
\usepackage{german}
\usepackage{amssymb}
\usepackage{amsfonts}
\usepackage{amsmath}
\usepackage{fullpage}
\usepackage{cite}
\usepackage{vmargin}

\usepackage{todonotes}

\newtheorem{theorem}{Theorem}

\newtheorem{definition}{Definition}

\newtheorem{lemma}{Lemma}

\newenvironment{proof}[1][Proof]{\noindent\textbf{#1.} }{\ \rule{0.5em}{0.5em}}

\setmarginsrb{3.5cm}{2.25cm}{3.5cm}{3.05cm}{0.3cm}{0.3cm}{-0.3cm}{1.0cm}

\usepackage{xspace}

\newcommand{\mcalg}{\ensuremath{\mathcal{G}}\xspace}
\newcommand{\mcal}{\ensuremath{\mathcal{M}}\xspace}
\newcommand{\mrgma}{\ensuremath{\mathcal{M}_{RG}}\xspace}

\newcommand{\map}{\ensuremath{\mathcal{M}^+}\xspace}
\newcommand{\mbad}{\ensuremath{\mathcal{M}^b}\xspace}
\newcommand{\man}{\ensuremath{\mathcal{M}^-}\xspace}
\newcommand{\mann}{\ensuremath{\mathcal{M}^{--}}\xspace}

\newcommand{\greedy}{\textsc{GreedyMatching}\xspace}
\newcommand{\greedyv}{\textsc{GreedyVertex}\xspace}
\newcommand{\greedye}{\textsc{GreedyEdge}\xspace}

\newcommand{\opt}{\textsc{Opt}\xspace}
\newcommand{\av}{\textsc{Av}\xspace}
\newcommand{\rgma}{\textsc{Rgma}\xspace}
\newcommand{\mrg}{\textsc{Mrg}\xspace}

\newcommand{\eps}{\epsilon\xspace}

\usepackage{pgf,pgfarrows,pgfnodes}
\usepackage{tikz}
\usetikzlibrary{arrows,positioning, shapes}
\tikzstyle{edge} = [draw,thick,-]
\tikzstyle{weight} = [font=\small]
\tikzstyle{matched} = [draw,line width=5pt,-,black!80]

\usepackage{tcolorbox}

\begin{document}

\title{\vspace{-0.5cm}The Complexity of Weighted Greedy Matching}

\author{Argyrios Deligkas\thanks{%
Department of Computer Science, University of Liverpool, UK. Email: \texttt{a.deligkas@liverpool.ac.uk}} 
\and George B.~Mertzios\thanks{%
School of Engineering and Computing Sciences, Durham University, UK. Email: \texttt{george.mertzios@durham.ac.uk}} 
\and Paul G. Spirakis\thanks{%
Department of Computer Science, University of Liverpool, UK, and Computer Technology Institute (CTI), Greece. Email: \texttt{p.spirakis@liverpool.ac.uk}}}
\date{\vspace{-1.0cm}}

\maketitle

\begin{abstract}
Motivated by the fact that in several cases a matching in a graph is stable if 
and only if it is produced by a greedy algorithm, we study the problem of 
computing a \emph{maximum weight greedy matching} on weighted graphs, termed \greedy. 
In wide contrast to the maximum weight matching problem, for which many efficient
algorithms are known, we prove that \textsc{GreedyMatching} is \emph{strongly NP-hard} and \emph{APX-complete}, and thus it does not admit a PTAS unless P=NP, 
even on graphs with maximum degree at most~3 and with at most three different integer edge weights. 
Furthermore we prove that \greedy is \emph{strongly NP-hard} if the input graph is in addition \emph{bipartite}.
Moreover we consider two natural parameters of the problem, for which we establish a \emph{sharp threshold} behavior between NP-hardness and tractability. 
On the positive side, we present a randomized approximation algorithm (\rgma) for \greedy on a special class of weighted graphs, called \emph{bush graphs}. 
We highlight an unexpected connection between \rgma and the approximation of 
maximum cardinality matching in unweighted graphs via randomized greedy algorithms.
We show that, if the approximation ratio of \textsc{Rgma} is $\rho$, 
then for every $\epsilon >0$ the randomized \textsc{Mrg} algorithm of~\cite{ADFS95} gives a $(\rho-\epsilon)$-approximation for the maximum cardinality matching. 
We conjecture that a tight bound for $\rho $ is $\frac{2}{3}$; we prove our conjecture true for two subclasses of bush graphs. 
Proving a tight bound for the approximation ratio of \textsc{Mrg} on unweighted graphs (and thus also proving a tight value for $\rho $) is a
long-standing open problem~\cite{PS12}. This unexpected relation
of our \textsc{Rgma} algorithm with the \textsc{Mrg} algorithm may provide
new insights for solving this problem.\newline

\end{abstract}

\section{Introduction}
The matching problem is one of the fundamental and most well studied problems in 
combinatorial optimization. Several different versions of the matching problem 
have been studied over the years: matchings on weighted or unweighted 
bipartite~\cite{FL10} and general~\cite{Edmonds} graphs, popular 
matchings~\cite{AIKM}, stable matchings~\cite{R84}, and greedy 
matchings~\cite{HK78}, to name but a few. 
In this article we investigate the problem of computing and approximating a 
maximum weight greedy matching on edge-weighted graphs; i.e.~the matching with 
the maximum weight when every edge that is added to the matching is chosen 
greedily from the set of available edges with the largest weight.

Although various polynomial time algorithms are known for the maximum weight 
matching problem, these algorithms are not always very fast in practice, 
see~\cite{Duan14, San09} and references therein for a comprehensive list of 
known results. 
One way to deal with this inefficiency is to turn into approximation algorithms;
a recent major result in this direction was given by Duan and Pettie~\cite{Duan14} 
who provided an algorithm for computing an $(1-\eps)$-approximation of the maximum 
weight matching in $O(m \eps^{-1} \log\eps^{-1})$ time, where $m$ is the number 
of edges. 
However, there are cases where it makes sense to search for algorithms that are 
fast in practice and easy to implement, such as \emph{greedy} algorithms. To make best 
use of such algorithms, their algorithmic performance needs to be properly understood.

There are cases where a greedy approach for computing a weighted matching is 
preferred or even \emph{required}, as the classical notion of a maximum weight 
matching does not necessarily fit the underlying problem. 
Consider for example the case where the vertices of a graph represent players 
and the edge weights represent the ``happiness'' that the corresponding 
players get from this match. 
It is possible that two players, which are not matched to each other in a 
maximum weighted matching, can still coordinate and match together instead of 
staying in the current matching, thus becoming both individually ``happier''. 
This is the so-called \emph{stable matching problem} which has received a 
lot of attention the previous years due to the plethora of its applications 
in real life problems, including kidney exchange~\cite{R04} and matching medical 
students to hospitals~\cite{RP99}. 
In many applications of the stable matching problem, such as in kidney exchange, 
there are only a few feasible values of ``happiness''. Thus in the underlying 
graph there are only few discrete edge weights, while many edges share the same 
weight. 
In such cases it is not clear how to compute a stable matching, as ties among 
edges with the same weight must be resolved.

It turns out that in many cases a matching is stable in the Nash equilibrium 
sense if and only if it can produced by a greedy algorithm~\cite{ADN13}. 
In the graph-matching game described above, a Nash equilibrium is a matching \mcal in which no two vertices $u,v$ can become individually happier 
by replacing their currently matched edges in \mcal with the edge $(u,v)$.
Thus, a \emph{maximum weight greedy} matching is an equilibrium (i.e.~stable matching) 
that maximizes the \emph{social welfare}, that is, the cumulative ``happiness'' of all the players.
A natural algorithmic question is whether a maximum weight greedy matching can be efficiently computed. 
Although greedy algorithms for matching problems have been studied extensively 
in the past~\cite{PS12, ADFS95, HK78, DF91, MP97, DFP93}, to the best of our knowledge
not much is known about the problem of computing a maximum weight greedy matching.

\subsection{Related work}
The scenarios of matching problems where the vertices of the graph correspond to 
players can vary from matching employees and employers~\cite{J79}, to matching
kidney donors and recipients~\cite{R04, ABS07}. Anshelevich, Das, and 
Naamad~\cite{ADN13} and Anshelevich, Bhardwaj, and Hoefer~\cite{ABH13} 
studied the price of anarchy and stability of stable matchings on weighted   
graphs. Furthermore, the authors in~\cite{ADN13} provided algorithms that 
compute almost stable matchings. Our work is closely related to~\cite{ADN13}, 
although their techniques cannot be applied to our problem since we study 
\emph{only} matchings that are greedy, whereas almost stable matchings are not.

Greedy matchings have been studied extensively over the years. 
The classical result 
by Korte and Hausmann~\cite{HK78} states that an arbitrary greedy matching is a 
$\frac{1}{2}$-approximation of the maximum cardinality matching, i.e.~every 
greedy matching on unweighted graphs picks at least half of maximum number of 
edges that any matching can pick. For edge-weighted graphs,
Avis~\cite{Avis83} showed that every algorithm that greedily picks edges with 
the maximum currently available weight is a $\frac{1}{2}$-approximation of the maximum 
weight matching. Hence, every weighted greedy matching is also a 
$\frac{1}{2}$-approximation for the maximum weight greedy matching problem. 
Several authors studied randomized greedy algorithms for the maximum cardinality
matching problem. The currently best randomized algorithm, known as \mrg~\cite{ADFS95}, picks the next edge to add to the matching 
by first selecting a random unmatched vertex $V$ of the graph and then a random 
unmatched neighbor of $v$. 
Aronson, Dyer, Frieze and Suen~\cite{ADFS95} showed that \mrg breaks the $\frac{1}{2}$-barrier 
and that it achieves a $\frac{1}{2}+ 1/400,000$-approximation guarantee on every graph. 
Recently, Poloczek and Szegedy~\cite{PS12} provided a different analysis for \mrg
and shown that it achieves an approximation guarantee of at least $\frac{1}{2}+
\frac{1}{256}$. However, as experiments suggest, the approximation guarantee of
\mrg can be as large as $\frac{2}{3}$~\cite{PS12}.

\subsection{Our contribution}
In this paper we study the computational complexity of computing and
approximating a \emph{maximum weight greedy matching} in a given
edge-weighted graph, i.e.~a greedy matching with the greatest weight among
all greedy matchings. This problem is termed \textsc{GreedyMatching}. In
wide contrast to the maximum weight matching, for which many efficient
algorithms are known (see \cite{Duan14} and the references therein), we prove
that \textsc{GreedyMatching} is \emph{strongly NP-hard} by a reduction from
a special case of MAX2SAT. Our reduction also implies hardness of
approximation; we prove that \textsc{GreedyMatching} is \emph{APX-complete},
and thus it does not admit a PTAS unless P=NP. These hardness results hold
even for input graphs with maximum degree at most 3 and with at most three
different integer edge weights, namely with weights in the set $\{1,2,4\}$.
Furthermore, by using a technique of Papadimitriou and Yannakakis~\cite{PY91}%
, we extend the NP-hardness proof to the interesting case where the input
graph is in addition \emph{bipartite}. Next, we study the decision
variations \textsc{GreedyVertex} and \textsc{GreedyEdge} of the problem,
where we now ask whether there exists a greedy matching in which a specific
vertex $u$ or a specific edge $(u,v)$ is matched. These are both natural
questions, as the designer of the stable matching might want to ensure
that a specific player or a specific pair of players is matched in the
solution. We prove that both \textsc{GreedyVertex} and \textsc{GreedyEdge}
are also strongly \textsc{NP}-hard.

As \textsc{GreedyMatching} turns out to be computationally hard, it makes
sense to investigate how the complexity is affected by appropriately
restricting the input. In this line of research we consider two natural
parameters of the problem, for which we establish a \emph{sharp threshold} behavior. As the first parameter we consider the \emph{minimum
ratio }$\lambda _{0}$ of any two \emph{consecutive weights}. Assume that the
graph has $\ell $ different edge weights $w_{1}>w_{2}>\ldots >w_{\ell }$; we
define for every $i\in \lbrack \ell -1]$ the ratio $\lambda _{i}=\frac{w_{i}%
}{w_{i+1}}$ and the minimum ratio $\lambda _{0}=\min_{i\in \lbrack \ell
-1]}\lambda _{i}$. We prove that, if $\lambda _{0}\geq 2$ then \textsc{%
GreedyMatching} can be solved in polynomial time, while for any
constant $\lambda _{0}<2$ \textsc{GreedyMatching} is strongly NP-hard and
APX-complete, even on graphs with maximum degree at most 3 and with at most
three different edge weights.

As the second parameter we consider the \emph{maximum edge cardinality} $\mu 
$ of the \emph{connected components} of $G(w_{i})$, among all different
weights $w_{i}$, where $G(w_{i})$ is the subgraph of $G$ spanned by the edges
of weight $w_{i}$. Although at first sight this parameter may seem
unnatural, it resembles the number of times that the greedy algorithm has to
break ties. At the stage where we have to choose among all available edges
of weight $w_{i}$, it suffices to consider each connected component of the
available edges of $G(w_{i})$ separately from the other components. In
particular, although the weight of the final greedy matching may highly
depend on the order of the chosen edges within a connected component, it is
independent of the ordering that the various different connected components
are processed. Thus $\mu $ is a reasonable parameter for \textsc{%
GreedyMatching}. In the case $\mu =1$ there exists a unique greedy matching
for $G$ which can be clearly computed in polynomial time. We prove that 
\textsc{GreedyMatching} is strongly NP-hard and APX-complete for $\mu \geq 2$%
, even on graphs with maximum degree at most 3 and with at most five
different edge weights.

On the positive side, we consider a special class of weighted graphs, called \emph{bush graphs}, where all edges of the same weight in $G$ form a star (bush).
We present a randomized approximation algorithm (\rgma) for \greedy on bush graphs and 
we highlight an unexpected connection between \rgma and the randomized \mrg algorithm for greedily approximating the maximum cardinality matching on unweighted graphs. 
In particular we show that, if the approximation ratio of \textsc{Rgma} for \textsc{GreedyMatching} on bush graphs is $\rho$, 
then for every $\epsilon >0$ \textsc{Mrg}~\cite{ADFS95} is a $(\rho-\epsilon )$-approximation algorithm for the maximum cardinality matching. 
We conjecture that a tight bound for $\rho $ is $\frac{2}{3}$; among our results we prove our conjecture true for two subclasses of bush graphs. 
Proving a tight bound for the approximation ratio of \textsc{Mrg} on unweighted graphs (and thus also proving a tight value for $\rho $) is a
long-standing open problem~\cite{PS12,ADFS95,DF91}. This unexpected relation
of our \textsc{Rgma} algorithm with the \textsc{Mrg} algorithm may provide
new insights for solving this problem.

\section{Preliminaries}
\label{sec:pre}
Every graph considered in this paper is undirected. For any graph 
$G = (V, E)$ we use $G + u$ to denote the graph $G' = (V', E')$ where 
$V' = V \cup \{u\}$ and $E'$ is consisted by the set $E$ and all the edges the 
vertex $u$ belongs to. Similarly $G - V'$ denotes the induced graph of $G$ 
defined by $V \setminus V'$, where $V'\subseteq V$. We study graphs $G = (V, E)$
with positive edge weights, i.e.~each edge $e=(u,v) \in E$ has a weight $w(e) = w(u,v) > 0$. 
The \emph{degree} of a vertex $u$ is the number of its adjacent vertices in $G$. 
We use $G(w_i)$ to denote the subgraph of $G$ spanned by the edges of weight $w_i$. 
A \emph{matching} $\mcal \subseteq E$ is a set of edges such that no pair of 
them are adjacent.  The weight of a matching \mcal is the sum of the weights of 
the edges in \mcal, formally $w(\mcal) = \sum_{e \in \mcal}w(e)$. A \emph{greedy 
matching} is a maximal matching constructed by the Greedy Matching Procedure.

\begin{tcolorbox}[title=Greedy Matching Procedure]
\textbf{Input:} Graph $G = (V, E)$, with $w_1 > w_2 > \ldots > w_{\ell}$ 
edge weight values\\
\textbf{Output:} Greedy matching \mcal
\begin{enumerate}
\item $\mcal \leftarrow \emptyset$
\item \textbf{for} $i = 1 \ldots \ell$ \textbf{do}
\item \hspace*{2mm} \textbf{while} there is an $e \in E$ such that $w(e) = w_i$ 
\textbf{do}
\item \hspace*{5mm} Pick an edge $e^* \in E$ with $w(e^*) = w_i$ and add it to \mcal;
\label{step4}
\item \hspace*{5mm} Remove all edges adjacent to $e^*$ from $E$;
\end{enumerate}
\end{tcolorbox}

\medskip

Notice that in Step~\ref{step4} 
the edge that is added to the matching \mcal is not specified explicitly.
The rule that specifies which edge is chosen in Step~\ref{step4} can be 
deterministic or randomized, resulting to a specific \emph{greedy matching algorithm}. 
We will use $\opt(G)$ to denote the optimum of \greedy with input $G$, i.e.~the weight of the \emph{maximum greedy matching} of $G$. 

\medskip

\begin{tcolorbox}[title=\greedy]
\textsc{Instance}: Graph $G=(V,E)$ with positive edge weights.\\
\textsc{Task}: Compute a maximum weight greedy matching \mcal for $G$.
\end{tcolorbox}

\medskip

Furthermore, we study another two related problems, where we ask whether there 
is a greedy matching that matches a specific vertex or a specific edge.

\medskip

\begin{tcolorbox}[title=\greedyv]
\textsc{Instance}: Graph $G=(V,E)$ with positive edge weights and a vertex ${v\in V}$.\\
\textsc{Question}: Is there a greedy matching \mcal such that $(v,u) \in \mcal$,
for some ${u \in V}$?
\end{tcolorbox}

\medskip

\begin{tcolorbox}[title=\greedye]
\textsc{Instance}: Graph $G=(V,E)$ with positive edge weights and an edge ${(u,v) \in E}$.\\
\textsc{Question}: Is there a greedy matching \mcal such that ${(v,u) \in \mcal}$?
\end{tcolorbox}

\section{Computational hardness of \greedy\label{sec:hardness}}
 In this section we study the complexity of computing a
maximum weight greedy matching. In Section~\ref{APX-subsec} we prove that 
\greedy is strongly NP-hard and
APX-complete, even on graphs with maximum degree at most 3 and with at most
three different integer weight values. By slightly modifying our reduction
of Section~\ref{APX-subsec}, we first prove in Section~\ref%
{hardness-bipartite-subsec} that \greedy remains
strongly NP-hard also when the graph is in addition bipartite, and we then
prove in Section~\ref{hardness-additional-subsec} that also the two decision problem
variations \greedyv and \greedye are
also strongly NP-hard. Our hardness reductions are from the MAX2SAT(3)
problem~\cite{Ausiello1999,RRR98}, which is the special case of MAX-SAT where in
the input CNF formula $\phi $ every clause has at most 2 literals and every
variable appears in at most 3 clauses; we call such a formula $\phi $ a
2SAT(3) formula.

Note that the decision version of \greedy, where we ask whether there exists a 
greedy matching with weight at least $B$, belongs to the class NP. Indeed we are
able to verify in polynomial time whether a given matching $\mcal$ is maximal, 
greedy and has weight at least $B$. The maximality and the weight of 
the matching $\mcal$ can be computed and checked in linear time. 
To check whether $\mcal$ is greedy, we first check whether the largest edge 
weight in $\mcal$ equals the largest edge weight in $G$. In this case we remove 
from $G$ all vertices incident to the highest weight edges of $\mcal$ and we 
apply recursively the same process in the resulting induced subgraph. Then $\mcal$ is greedy if and only if we end up with a graph with no 
edges.

\subsection{Overview of the reduction.\label{overview-subsec}}
Given a 2SAT(3) formula $\phi $ with $m$ clauses and $n$ variables $%
x_{1},\ldots ,x_{n}$ we construct an undirected graph $G$ with $10n+m$
vertices and $9n+2m$ edges. Then we prove that there exists a truth
assignment that satisfies at least $k$ clauses of $\phi $ if and only if
there exists a greedy matching \mcal in $G$ with weight at
least $14n+k$. Without loss of generality we make the following assumptions
on $\phi $. Firstly, if a variable occurs only with positive (resp.~only
with negative) literals, then we trivially set it true (resp.~false) and
remove the associated clauses. Furthermore, without loss of generality, if a
variable $x_{i}$ appears three times in $\phi $, we assume that it appears
once as a positive literal $x_{i}$ and two times as a negative literal $%
\overline{x_{i}}$; otherwise we rename the negation with a new variable.
Similarly, if $x_{i}$ appears two times in $\phi $, then it appears once as
a positive literal $x_{i}$ and once as a negative literal $\overline{x_{i}}$.

For each variable $x_{i}$ we create a subgraph $\mcalg_{x_{i}}$
and for each clause $C_{j}$ we create one vertex $v_{j}$. The vertices
created from the clauses will be called $v$-vertices. Each subgraph $%
\mcalg_{x_{i}}$ is a path with 10 vertices, where three of them
are distinguished; the vertices $\alpha _{x_{i}},\beta _{x_{i}}$ and $\gamma
_{x_{i}}$. Each distinguished vertex can be connected with at most one $v$%
-vertex that represents a clause. Furthermore, every $v$-vertex is connected
with at most two vertices from the subgraphs $\mcalg_{x_{i}}$;
one distinguished vertex from each of the subgraphs $\mcalg%
_{x_{i}}$ that correspond to the variables of the clause. The edge weights
in the subgraphs $\mcalg_{x_{i}}$ are not smaller than the
weights of the edges connecting the $v$-vertices with the distinguished
vertices of the subgraphs $\mcalg_{x_{i}}$.

\subsection{The construction\label{construction-subsec}} 
The gadget $\mcalg_{x_i}$ that we create for variable $x_i$ is illustrated in 
Figure~\ref{fig:gx}; the distinguished vertices of $\mcalg_{x_i}$ are 
$\alpha_{x_i}$, $\beta_{x_i}$ and $\gamma_{x_i}$. 
The vertex $\alpha_{x_i}$ corresponds to the \emph{positive} literal of the 
variable and vertices $\beta_{x_i}$ and $\gamma_{x_i}$ correspond to the 
\emph{negative} literal~$\overline{x_i}$. 
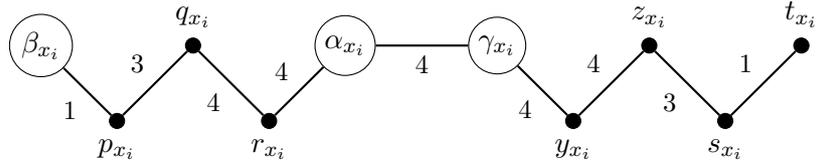
\begin{figure}[h!]
\label{fig:gx}
\begin{center}
\begin{tikzpicture}[scale=1, auto,swap]

\node[draw,circle,inner sep=2pt] (01) at (0,1) {$\beta_{x_i}$};
\node[draw,circle,inner sep=2pt,fill] (10) at (1,0) [label=below:$p_{x_i}$]{};
\path[edge] (01) -- node[weight] {$1$} (10);
\node[draw,circle,inner sep=2pt,fill] (21) at (2,1) [label=above:$q_{x_i}$]{};
\path[edge] (21) -- node[weight] {$3$} (10);
\node[draw,circle,inner sep=2pt,fill] (30) at (3,0) [label=below:$r_{x_i}$]{};
\path[edge] (21) -- node[weight] {$4$} (30);
\node[draw,circle,inner sep=2pt] (41) at (4,1) {$\alpha_{x_i}$};
\path[edge] (41) -- node[weight] {$4$} (30);
\node[draw,circle,inner sep=2pt] (51) at (6,1) {$\gamma_{x_i}$};
\path[edge] (41) -- node[weight] {$4$} (51);
\node[draw,circle,inner sep=2pt,fill] (70) at (7,0) [label=below:$y_{x_i}$]{};
\path[edge] (51) -- node[weight] {$4$} (70);
\node[draw,circle,inner sep=2pt,fill] (81) at (8,1) [label=above:$z_{x_i}$]{};
\path[edge] (81) -- node[weight] {$4$} (70);
\node[draw,circle,inner sep=2pt,fill] (90) at (9,0) [label=below:$s_{x_i}$]{};
\path[edge] (81) -- node[weight] {$3$} (90);
\node[draw,circle,inner sep=2pt,fill] (101) at (10,1) [label=above:$t_{x_i}$]{};
\path[edge] (101) -- node[weight] {$1$} (90);
\end{tikzpicture}

\caption{The gadget $\mcalg_{x_i}$.}
\end{center}
\end{figure}

The vertex $v_{j}$ associated to clause $C_{j}$, where $j\in \lbrack m]$, is
made adjacent to the vertices that correspond to the literals associated
with that clause. For example, if $C_{j}=(x_{1}\vee \overline{x_{2}})$ we
will connect the vertex $v_{j}$ with one of the vertices $\alpha
_{x_{1}},\beta _{x_{1}},\gamma _{x_{1}}$ and with one of the vertices $%
\alpha _{x_{2}},\beta _{x_{2}},\gamma _{x_{2}}$. In order to make these
connections in a consistent way, we first fix an arbitrary ordering over the
clauses. If the variable $x_{i}$ occurs as a positive literal in the clause $%
C_{j}$, then we add the edge $(v_{j},\alpha _{x_{i}})$ of weight 3. Next, if 
$C_{j}$ is the first clause that the variable $x_{i}$ occurs with a negative
literal (in the fixed ordering of the clauses), then we add the edge $%
(v_{j},\beta _{x_{i}})$ of weight 1. Finally, if the clause $C_{j}$ is the
second clause that the variable $x_{i}$ occurs as a negative literal, then
we add the edge $(v_{j},\gamma _{x_{i}})$ of weight 3. That is, if a
variable $x_{i}$ appears only two times in $\phi $, then only the two
distinguished vertices $\alpha _{x_{i}}$ and $\beta _{x_{i}}$ of $\mcalg_{x_{i}}$ are adjacent to a $v$-vertex. This completes the
construction of the graph $G$. Note that, by the construction of $G$, in any
maximum greedy matching of $G$, there are exactly four alternative ways to
match the edges of each of the subgraphs $\mcalg_{x_{i}}$, as
illustrated in Fig.~\ref{fig:man}-\ref{fig:bad2}.

\subsection{APX-completeness\label{APX-subsec}} 
  
In order to prove that \greedy is APX-complete, first we prove in the next lemma that given an assignment that satisfies at least $k$ 
clauses we can construct a greedy matching with weight at least $14n+k$. 
The intuition for this lemma is as follows. Starting with a given satisfying 
truth assignment~$\tau$ for the input formula~$\phi$, we first construct the 
matching $\man$ in every $\mcalg_{x_i}$ (cf.~Figure~\ref{fig:man}), and thus the
$\beta$-vertices are initially free to be matched to $v$-vertices. Then, if a 
variable $x_i$ is true in $\tau$, we change the matching of $\mcalg_{x_i}$ from 
$\man$ to $\map$ (cf.~Figure~\ref{fig:map}), such that only the $\alpha$-vertex 
(and not the $\beta$ and $\gamma$-vertices) of $\mcalg_{x_i}$ is free to be 
matched to a $v$-vertex. On the other hand, if the variable $x_i$ is false in 
$\tau$, then we either keep the matching $\man$ in $\mcalg_{x_i}$, or we replace
$\man$ with $\mann$ in $\mcalg_{x_i}$ (cf.~Figure~\ref{fig:mann}). Note that in 
$\man$ only $\beta_{x_i}$ is free to be matched, while in $\mann$ both 
$\beta_{x_i}$ and $\gamma_{x_i}$ are free to be matched with a $v$-vertex; in 
both cases the $\alpha$-vertex of $\mcalg_{x_i}$ is ``blocked'' from being 
matched to a $v$-vertex. Then, using the fact that $\tau$ satisfies at least $k$ clauses 
of $\phi$, we can construct a matching of~$G$ where $k$ $v$-vertices are matched 
and the total weight of this matching is at least~${14n+k}$.

\begin{lemma}
\label{lem:one}
If there is an assignment that satisfies at least $k$ clauses then, there is a greedy 
matching with weight at least $14n+k$.
\end{lemma}

\begin{proof}
Given an assignment that satisfies $k$ clauses we will construct a greedy 
matching of $G$ with weight $14n+k$ by making use of the three matchings $\man$,
$\mann$, and $\map$ of $\mcalg_{x_i}$, as illustrated in 
Figures~\ref{fig:man}-\ref{fig:map}. 
All these three matchings of $\mcalg_{x_i}$ are greedy; furthermore note that
there also exists a fourth greedy matching \mbad of $\mcalg_{x_i}$ 
(see Fig.~\ref{fig:bad2}) which will not be used in the proof of the lemma.

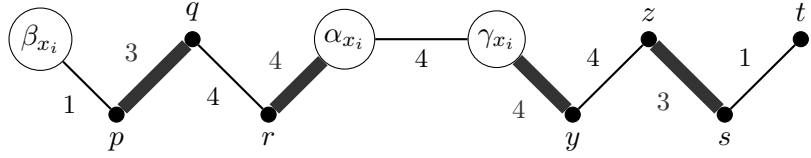
\begin{figure}[h!]
\begin{center}
\begin{tikzpicture}[scale=1, auto,swap]
\node[draw,circle,inner sep=2pt] (01) at (0,1) {$\beta_{x_i}$};
\node[draw,circle,inner sep=2pt,fill] (10) at (1,0) [label=below:$p$]{};
\path[edge] (01) -- node[weight] {$1$} (10);
\node[draw,circle,inner sep=2pt,fill] (21) at (2,1) [label=above:$q$]{};
\path[matched] (21) -- node[weight] {$3$} (10);
\node[draw,circle,inner sep=2pt,fill] (30) at (3,0) [label=below:$r$]{};
\path[edge] (21) -- node[weight] {$4$} (30);
\node[draw,circle,inner sep=2pt] (41) at (4,1) {$\alpha_{x_i}$};
\path[matched] (41) -- node[weight] {$4$} (30);
\node[draw,circle,inner sep=2pt] (51) at (6,1) {$\gamma_{x_i}$};
\path[edge] (41) -- node[weight] {$4$} (51);
\node[draw,circle,inner sep=2pt,fill] (70) at (7,0) [label=below:$y$]{};
\path[matched] (51) -- node[weight] {$4$} (70);
\node[draw,circle,inner sep=2pt,fill] (81) at (8,1) [label=above:$z$]{};
\path[edge] (81) -- node[weight] {$4$} (70);
\node[draw,circle,inner sep=2pt,fill] (90) at (9,0) [label=below:$s$]{};
\path[matched] (81) -- node[weight] {$3$} (90);
\node[draw,circle,inner sep=2pt,fill] (101) at (10,1)[label=above:$t$] {};
\path[edge] (101) -- node[weight] {$1$} (90);
\end{tikzpicture}
\caption{The matching $\man$ with weight 14 for the subgraph $\mcalg_{x_i}$. 
For simplicity of notation we do not include the subscript $x_i$ in the non-distinguished vertices.}
\label{fig:man}
\end{center}
\end{figure}

\begin{figure}[h!]
\begin{center}
\begin{tikzpicture}[scale=1, auto,swap]
\node[draw,circle,inner sep=2pt] (01) at (0,1) {$\beta_{x_i}$};
\node[draw,circle,inner sep=2pt,fill] (10) at (1,0) [label=below:$p$]{};
\path[edge] (01) -- node[weight] {$1$} (10);
\node[draw,circle,inner sep=2pt,fill] (21) at (2,1) [label=above:$q$]{};
\path[matched] (21) -- node[weight] {$3$} (10);
\node[draw,circle,inner sep=2pt,fill] (30) at (3,0) [label=below:$r$]{};
\path[edge] (21) -- node[weight] {$4$} (30);
\node[draw,circle,inner sep=2pt] (41) at (4,1) {$\alpha_{x_i}$};
\path[matched] (41) -- node[weight] {$4$} (30);
\node[draw,circle,inner sep=2pt] (51) at (6,1) {$\gamma_{x_i}$};
\path[edge] (41) -- node[weight] {$4$} (51);
\node[draw,circle,inner sep=2pt,fill] (70) at (7,0) [label=below:$y$]{};
\path[edge] (51) -- node[weight] {$4$} (70);
\node[draw,circle,inner sep=2pt,fill] (81) at (8,1) [label=above:$z$]{};
\path[matched] (81) -- node[weight] {$4$} (70);
\node[draw,circle,inner sep=2pt,fill] (90) at (9,0) [label=below:$s$] {};
\path[edge] (81) -- node[weight] {$3$} (90);
\node[draw,circle,inner sep=2pt,fill] (101) at (10,1) [label=above:$t$]{};
\path[matched] (101) -- node[weight] {$1$} (90);
\end{tikzpicture}
\caption{The matching $\mann$ with weight 12 for the subgraph $\mcalg_{x_i}$.}
\label{fig:mann}
\end{center}
\end{figure}

\begin{figure}[h!]
\begin{center}
\begin{tikzpicture}[scale=1, auto,swap]
\node[draw,circle,inner sep=2pt] (01) at (0,1) {$\beta_{x_i}$};
\node[draw,circle,inner sep=2pt,fill] (10) at (1,0) [label=below:$p$]{};
\path[matched] (01) -- node[weight] {$1$} (10);
\node[draw,circle,inner sep=2pt,fill] (21) at (2,1) [label=above:$q$]{};
\path[edge] (21) -- node[weight] {$3$} (10);
\node[draw,circle,inner sep=2pt,fill] (30) at (3,0) [label=below:$r$]{};
\path[matched] (21) -- node[weight] {$4$} (30);
\node[draw,circle,inner sep=2pt] (41) at (4,1) {$\alpha_{x_i}$};
\path[edge] (41) -- node[weight] {$4$} (30);
\node[draw,circle,inner sep=2pt] (51) at (6,1) {$\gamma_{x_i}$};
\path[edge] (41) -- node[weight] {$4$} (51);
\node[draw,circle,inner sep=2pt,fill] (70) at (7,0) [label=below:$y$]{};
\path[matched] (51) -- node[weight] {$4$} (70);
\node[draw,circle,inner sep=2pt,fill] (81) at (8,1) [label=above:$z$]{};
\path[edge] (81) -- node[weight] {$4$} (70);
\node[draw,circle,inner sep=2pt,fill] (90) at (9,0) [label=below:$s$]{};
\path[matched] (81) -- node[weight] {$3$} (90);
\node[draw,circle,inner sep=2pt,fill] (101) at (10,1) [label=above:$t$]{};
\path[edge] (101) -- node[weight] {$1$} (90);
\end{tikzpicture}
\caption{The matching $\map$ with weight 12 for the subgraph $\mcalg_{x_i}$.}
\label{fig:map}
\end{center}
\end{figure}

We construct the greedy matching of $G$ with weight $14n+k$ as follows. 
Firstly, we set all the matchings for the subgraphs $\mcalg_{x_i}$ to 
be the greedy matching $\man$, thus incurring a total weight of $14n$ from the 
currently matched edges. 
Then we process sequentially each clause $C_j$ of the formula $\phi$. 
If a clause $C_j$ is satisfied in the given truth assignment by at least one 
\emph{positive} literal, then we choose one of these literals arbitrarily, say 
$x_i$, and we change the matching of $\mcalg_{x_i}$ to $\map$; 
furthermore we match the edge $(v_{j},\alpha_{x_i})$ which has weight 3. 
In this case we replaced the matched edges $(p_{x_i},q_{x_i})$ and 
$(r_{x_i},\alpha_{x_i})$ of $\mcalg_{x_i}$ with total weight 7 by the matched 
edges $(\beta_{x_i},p_{x_i})$, $(q_{x_i},r_{x_i})$, and $(v_j,\alpha_{x_i})$ 
with total weight 8, i.e.~we increased the weight of the matching by 1.

Assume that a clause $C_j$ is satisfied in the given assignment only by 
\emph{negative} literals. 
If at least one of these literals of $C_j$ corresponds to a $\beta$-vertex, then
we match the edge $(v_{j},\beta_{x_i})$ of weight 1. Thus in this case we also 
increase the total weight of the matched edges by 1. 
Finally, if all of these literals of $C_j$ correspond to $\gamma$-vertices, then
we choose one of them arbitrarily, say $\overline{x_i}$, and we change the matching 
of $\mcalg_{x_i}$ to $\mann$; furthermore we match the edge $(v_{j},\gamma_{x_i})$ 
of weight 3. 
In this case we replaced the matched edges $(\gamma_{x_i},y_{x_i})$ and 
$(z_{x_i},s_{x_i})$ of $\mcalg_{x_i}$ with total weight 7 by the matched edges 
$(v_j,\gamma_{x_i})$, $(y_{x_i},z_{x_i})$, and $(s_{x_i},t_{x_i})$ with total 
weight 8, i.e.~we increased the weight of the matching by 1. 
This completes the required matching $\mathcal{M}$ of the graph $G$.

Since we started with a matching of total weight $14n$ and we added weight 1 for
each of the $k$ satisfied clauses in $\phi$, note that the total weight of 
$\mathcal{M}$ is $14n+k$. 
In this matching $\mathcal{M}$, each of the induced subgraphs $\mcalg_{x_i}$ of 
$G$ is greedily matched. 
Furthermore all the remaining edges of $G$ are edges that join a $v$-vertex with
a distinguished vertex $\alpha_{x_i}$ (resp.~$\beta_{x_i}$, $\gamma_{x_i}$). 
Note that the weight of each of these edges is smaller than or equal to the 
weight of the edges adjacent to $\alpha_{x_i}$ (resp.~$\beta_{x_i}$, $\gamma_{x_i}$)
within the subgraph $\mcalg_{x_i}$. 
Thus, the matching $\mathcal{M}$ of $G$ can be constructed greedily. Moreover, 
since $\mathcal{M}$ can be potentially further extended greedily to a matching 
with larger weight, it follows that the maximum greedy matching of $G$ is at 
least $14n+k$.
\end{proof}

\medskip

Next we prove in Lemma~\ref{lem:at-least} that, if there is a greedy matching 
with weight $14n+k$, then there is an assignment that satisfies at least $k$ 
clauses.
In order to prove Lemma~\ref{lem:at-least}, first we prove in Lemma~\ref{lem:one-of-ag} a crucial property of the constructed graph $G$, 
namely that in any greedy matching at most one of the vertices $\alpha_{x_i}$ and $\gamma_{x_i}$ can be matched with a $v$-vertex.
\begin{lemma}
\label{lem:one-of-ag}
Let $\mathcal{M}$ be an arbitrary greedy matching of $G$ and let 
$i\in \{1,2,\ldots,n\}$. Then, in the subgraph $\mcalg_{x_i}$, at most one of 
the vertices $\alpha_{x_i}$ and $\gamma_{x_i}$ can be matched with a $v$-vertex.
\end{lemma}

\begin{proof}
The proof is done by contradiction. Assume otherwise that both $\alpha_{x_i}$ 
and $\gamma_{x_i}$ are matched with some $v$-vertices in $\mathcal{M}$. 
Note that both these edges that connect the vertices $\alpha_{x_i}$ and 
$\gamma_{x_i}$ with the corresponding $v$-vertices have weight~3. Furthermore, 
none of the edges $(\alpha_{x_i},r)$, $(\gamma_{x_i},y)$, and 
$(\alpha_{x_i},\gamma_{x_i})$ belong to $\mathcal{M}$. 
Thus, since the weight of the edge $(\alpha_{x_i},\gamma_{x_i})\notin 
\mathcal{M}$ is~4, it follows $\mathcal{M}$ is not greedy, which is a 
contradiction. 
That is, if both edges $(\alpha_{x_i},r)$ and $(\gamma_{x_i},y)$ of the subgraph
$\mcalg_{x_i}$ are not matched within $\mathcal{M}$, then 
$(\alpha_{x_i},\gamma_{x_i})\in \mathcal{M}$, as it is illustrated in the ``bad''
matching $\mathcal{M}^{b}$ of Fig.~\ref{fig:bad2}.
\end{proof}

\begin{figure}[h!]
\begin{center}
\begin{tikzpicture}[scale=1, auto,swap]
\node[draw,circle,inner sep=2pt] (01) at (0,1) {$\beta_{x_i}$};
\node[draw,circle,inner sep=2pt,fill] (10) at (1,0) [label=below:$p$]{};
\path[matched] (01) -- node[weight] {$1$} (10);
\node[draw,circle,inner sep=2pt,fill] (21) at (2,1) [label=above:$q$]{};
\path[edge] (21) -- node[weight] {$3$} (10);
\node[draw,circle,inner sep=2pt,fill] (30) at (3,0) [label=below:$r$]{};
\path[matched] (21) -- node[weight] {$4$} (30);
\node[draw,circle,inner sep=2pt] (41) at (4,1) {$\alpha_{x_i}$};
\path[edge] (41) -- node[weight] {$4$} (30);
\node[draw,circle,inner sep=2pt] (51) at (6,1) {$\gamma_{x_i}$};
\path[matched] (41) -- node[weight] {$4$} (51);
\node[draw,circle,inner sep=2pt,fill] (70) at (7,0) [label=below:$y$]{};
\path[edge] (51) -- node[weight] {$4$} (70);
\node[draw,circle,inner sep=2pt,fill] (81) at (8,1) [label=above:$z$]{};
\path[matched] (81) -- node[weight] {$4$} (70);
\node[draw,circle,inner sep=2pt,fill] (90) at (9,0) [label=below:$s$]{};
\path[edge] (81) -- node[weight] {$3$} (90);
\node[draw,circle,inner sep=2pt,fill] (101) at (10,1) [label=above:$t$] {};
\path[matched] (101) -- node[weight] {$1$} (90);
\end{tikzpicture}
\caption{The ``bad'' matching \mbad for the subgraph $\mcalg_{x_i}$ with weight 14.}
\label{fig:bad2}
\end{center}
\end{figure}

We are now ready to prove Lemma~\ref{lem:at-least}.

\begin{lemma}
\label{lem:at-least}
If there is a greedy matching with weight at least $14n+k$ in $G$, then there exists 
an assignment that satisfies at least $k$ clauses of the formula $\phi$.
\end{lemma}

\begin{proof}
Let \mcal be a maximum weight greedy matching of $G$ and assume that \mcal has 
weight at least $14n+k$. 
First we show that we can assume without loss of generality that, for every 
$i\in [n]$, the edges of the induced subgraph $\mcalg_{x_i}$ are matched in \mcal 
according to one of the four matchings $\man$, $\mann$, $\map$, and \mbad (see Figures~\ref{fig:man}-\ref{fig:bad2}). 
Assume that the edge $(\gamma_{x_i},y_{x_i})$ is matched in \mcal. 
Then clearly the edge $(z_{x_i},s_{x_i})$ is also matched in \mcal, since this is the only valid greedy option for the right part of $\mcalg_{x_i}$. 
Assume that the vertex $\gamma_{x_i}$ is matched in \mcal with a vertex different 
than $y_{x_i}$. Then similarly the edges $(y_{x_i},z_{x_i})$ and 
$(s_{x_i},t_{x_i})$ are matched in \mcal. On the other hand, assume that the 
edge $(\alpha_{x_i},r_{x_i})$ is matched in \mcal. 
Then the edge $(p_{x_i},q_{x_i})$ is also matched in \mcal due to the assumption
that \mcal is greedy.
Finally assume that the vertex $\alpha_{x_i}$ is matched in \mcal with a vertex
different than $r_{x_i}$. 
Then, since \mcal is greedy, the edge $(q_{x_i},r_{x_i})$ is matched in \mcal. 
Furthermore, since \mcal has the greatest weight among the greedy matchings of 
$G$ by assumption, the vertex $\beta_{x_i}$ is matched in \mcal either with 
vertex $p_{x_i}$ or with its adjacent $v$-vertex. 
If $\beta_{x_i}$ is matched with the $v$-vertex, then we replace this matched 
edge in \mcal by the matched edge $(\beta_{x_i},p_{x_i})$ (if the vertex
$p_{x_i}$ is unmatched) and we get an other greedy matching with the same weight. 
Therefore, we can assume without loss of generality that, for every $i\in [n]$, 
the edges of the induced subgraph $\mcalg_{x_i}$ are matched in \mcal 
according to one of the four matchings $\man$, $\mann$, $\map$, and \mbad, see 
Figures~\ref{fig:man}-\ref{fig:bad2}.

We construct from \mcal a truth assignment that satisfies at least $k$ clauses 
of the formula $\phi$, as follows. 
If the edges of the induced subgraph $\mcalg_{x_i}$ are matched in \mcal according to one of the matchings $\man$ or $\mann$, then we set the value of $x_i$ to \emph{false}. 
If the edges of $\mcalg_{x_i}$ are matched according to $\map$, then we set the value of $x_i$ to \emph{true}. 
Otherwise, if the edges of $\mcalg_{x_i}$ are matched according to \mbad, then set the truth value of $x_i$ arbitrarily. 
Let now $i \in [n]$. Since the edges of the induced subgraph $\mcalg_{x_i}$ are matched according to one of the matchings $\man$, $\mann$, $\map$, and \mbad, as we proved above, 
it follows that the vertices $\alpha_{x_i}$ and $\beta_{x_i}$ are not simultaneously matched with their associated $v$-vertices in \mcal. 
Furthermore, Lemma~\ref{lem:one-of-ag} implies that at most one of the vertices 
$\alpha_{x_i}$ and $\gamma_{x_i}$ can be matched with their associated $v$-vertices
in \mcal. 
Therefore the constructed truth assignment is valid.

Let $i \in [n]$. If $\mcalg_{x_i}$ is matched according to \mbad in \mcal, then 
$\mcalg_{x_i}$ clearly contributes weight 14 to the total weight of \mcal. 
Assume that $\mcalg_{x_i}$ is matched according to $\map$ in \mcal, i.e.~assume 
that $\alpha_{x_i}$ is matched with a $v$-vertex. 
Then $\gamma_{x_i}$ is matched with $y_{x_i}$, and thus the right part of 
$\mcalg_{x_i}$ (i.e.~the part between vertices $\gamma_{x_i}$ and $t_{x_i}$) 
contributes weight 7 to the total weight of \mcal. 
Furthermore the left part of $\mcalg_{x_i}$ (i.e.~the part between vertices 
$\beta_{x_i}$ and $\alpha_{x_i}$) contributes weight 1+4+3=8 to the total weight
of \mcal. 
That is, $\mcalg_{x_i}$ contributes weight 15 to the total weight of \mcal. 
Assume now that $\mcalg_{x_i}$ is matched according to $\man$ or $\mann$ in \mcal. 
If $\beta_{x_i}$ is matched with $p_{x_i}$ in \mcal, then the left part of 
$\mcalg_{x_i}$contributes weight 3+4=7 to the total weight of \mcal. 
Otherwise, if $\beta_{x_i}$ is matched with its adjacent $v$-vertex in \mcal, 
then the left part of $\mcalg_{x_i}$ contributes weight 1+3+4=8 to the total 
weight of \mcal. 
Similarly, if $\gamma_{x_i}$ is matched with $y_{x_i}$ in \mcal, then the right 
part of $\mcalg_{x_i}$ contributes weight 4+3=7 to the total weight of \mcal. 
Otherwise, if $\gamma_{x_i}$ is matched with its adjacent $v$-vertex in \mcal, 
then the right part of $\mcalg_{x_i}$ contributes weight 3+4+1=8 to the total 
weight of \mcal. 
Summarizing, if $0 \leq \ell \leq 2$ of the vertices $\{\alpha_{x_i}, 
\beta_{x_i}, \gamma_{x_i}\}$ are matched with $v$-vertices in \mcal, then 
$\mcalg_{x_i}$ contributes weight $14+\ell$ to the total weight of \mcal.

Therefore, since $G$ has $n$ induced subgraphs and the weight of \mcal is at 
least $14n+k$, it follows that $\ell\geq k$ $v$-vertices are matched in \mcal. 
For every vertex $v_j$ that is matched with a vertex $\alpha_{x_i}$ in \mcal, 
the clause $C_j$ contains the literal $x_i$ and the variable $x_i$ is set to 
true by the construction of the truth assignment. Thus $C_j$ is satisfied. 
Similarly, for every vertex $v_j$ that is matched with a vertex $\beta_{x_i}$
or $\gamma_{x_i}$ in \mcal, the clause $C_j$ contains the literal $\overline{x_i}$ 
and the variable $x_i$ is set to false by the construction of the truth 
assignment. Thus $C_j$ is again satisfied. 
Therefore, since $\ell\geq k$ $v$-vertices are matched in \mcal, it follows that
there are $\ell\geq k$ satisfied clauses of $\phi$ in the constructed assignment.
\end{proof}

\medskip

In the following theorem we conclude with the main result of this section.

\begin{theorem}
\label{thm:inai} \greedy is \emph{strongly NP-hard}
and \emph{APX-complete}. In particular, unless P=NP, \textsc{GreedyMatching}%
\xspace admits no PTAS, even on graphs with maximum degree at most 3 and with
at most three different integer weight values.
\end{theorem}

\begin{proof}
It follows by Lemmas~\ref{lem:one} and~\ref{lem:at-least} that there is a
greedy matching \mcal in $G$ with weight at least $14n+k$ if
and only if there is a truth assignment that satisfies at least $k$ clauses
in the 2SAT(3) formula $\phi $. Thus it follows that \textsc{GreedyMatching} is NP-hard,
since MAX2SAT(3) is also NP-hard~\cite{Ausiello1999,RRR98}. Furthermore, since the graph $G$ has three
different weight values (namely 1, 3, and 4), it follows that \textsc{%
GreedyMatching} is strongly NP-hard.

Denote by OPT$_{\text{Max2SAT(3)}}(\phi )$ the greatest number of clauses
that can be simultaneously satisfied by a truth assignment of $\phi $.
Furthermore denote by OPT$_{\text{Greedy}}(G)$ the maximum weight of a
greedy matching of the graph $G$ that is constructed from $\phi $ by our
reduction. Recall by construction that $G$ has 3 different integer weights
and the maximum degree is 3. Then Lemma~\ref{lem:one} implies that OPT$_{%
\text{Greedy}}(G)\geq 14n+$OPT$_{\text{Max2SAT(3)}}(\phi )$. Note that a
random truth assignment satisfies each clause of $\phi $ with probability $%
\frac{7}{8}$, and thus there exists a truth assignment that satisfies at
least $\frac{7}{8}m$ clauses of $\phi $, where $m$ is the number of clauses,
and thus OPT$_{\text{Max2SAT(3)}}(\phi )\geq \frac{7}{8}m$. Since every
clause has at most 2 literals in $\phi $, it follows that $m\geq \frac{n}{2}$%
, and thus OPT$_{\text{Max2SAT(3)}}(\phi )\geq \frac{7}{8}m\geq \frac{7}{16}n
$.

Assume that there is a PTAS for computing OPT$_{\text{Greedy}}(G)$. Then for
every $\varepsilon >0$ we can compute in polynomial time a greedy matching $%
\mathcal{M}$ of $G$ such that $|\mathcal{M}|\geq (1-\varepsilon )\cdot $OPT$%
_{\text{Greedy}}(G)$. Given such a matching $\mathcal{M}$ we can compute by
Lemma \ref{lem:at-least} a truth assignment $\tau $ of $\phi $ such that $%
|\tau (\phi )|\geq |\mathcal{M}|-14n$. Therefore:%
\begin{eqnarray*}
|\tau (\phi )| &\geq &(1-\varepsilon )\cdot OPT_{\text{Greedy}}(G)-14n \\
&\geq &(1-\varepsilon )\cdot (14n+OPT_{\text{Max2SAT(3)}}(\phi ))-14n \\
&\geq &(1-\varepsilon )\cdot OPT_{\text{Max2SAT(3)}}(\phi )-14\varepsilon
\cdot \frac{16}{7}OPT_{\text{Max2SAT(3)}}(\phi ) \\
&\geq &(1-33\varepsilon )\cdot OPT_{\text{Max2SAT(3)}}(\phi )
\end{eqnarray*}

That is, assuming a PTAS for computing OPT$_{\text{Greedy}}(G)$, we obtain a
PTAS for computing OPT$_{\text{Max2SAT(3)}}(\phi )$. This is a contradiction
by \cite{Ausiello1999}, unless P=NP. This proves that \textsc{GreedyMatching}%
\xspace is APX-hard. Furthermore \greedy clearly
belongs to the class APX, as any greedy matching algorithm achieves an $%
\frac{1}{2}$-approximation for \greedy, and thus 
\greedy is APX-complete.
\end{proof}

\subsection{Hardness of \greedy in Bipartite graphs
\label{hardness-bipartite-subsec}}

The graph $G$ that we constructed from $\phi$ (see Section~\ref%
{construction-subsec}) is not necessarily bipartite, as it may contain an
odd-length cycle. More specifically, it is possible that the following cycle
of length 9 exists: 
$$
v \rightarrow \beta_{x_i} \rightarrow p \rightarrow q \rightarrow r
\rightarrow \alpha_{x_i} \rightarrow v^{\prime }\rightarrow \gamma_{x_j}
\rightarrow \alpha_{x_j} \rightarrow v.
$$
However, as we prove in this section, \greedy remains
strongly NP-hard also when the graph is in addition
bipartite.

To prove this (cf.~Theorem~\ref{thm:bip-hard}), we slightly modify our
reduction of Section~\ref{construction-subsec} and the proofs of Section \ref%
{APX-subsec}, as follows. We start with a 2-CNF formula $\phi $, where every
variable appears in an \emph{arbitrary} number of clauses. We may assume
without loss of generality that every variable appears in $\phi $ at least
three times; otherwise we may add dummy copies of existing clauses. Then we
create from $\phi $ an equivalent 2-CNF formula $\phi ^{\prime }$ using a
technique of Papadimitriou and Yannakakis~\cite{PY91}. More specifically,
for every variable $x_{i}$ that appears $l\geq 3$ times in $\phi $, we
replace $x_{i}$ by $l$ new variables $x_{i_{1}},\ldots ,x_{i_{l}}$, one for
every clause of $\phi _{1}$ in which $x_{i}$ initially appeared. Furthermore
we add the $l$ extra clauses $(\overline{x_{i_{1}}}\vee x_{i_{2}})$, $(%
\overline{x_{i_{2}}}\vee x_{i_{3}})$, $\ldots $, $(\overline{x_{i_{l}}}\vee
x_{i_{1}})$. Denote by $\phi ^{\prime }$ the resulting 2-CNF formula after
performing these operations for every $i=1,2,\ldots ,n$. Note that in $\phi
^{\prime }$ a variable $x_{i_{j}}$ occurs exactly in three clauses: two
times as $\overline{x_{i_{j}}}$ and one as $x_{i_{j}}$ if $x_{i_{j}}$ was
negative in $\phi $, or two times as $x_{i_{j}}$ and one as $\overline{%
x_{i_{j}}}$ if $x_{i_{j}}$ was positive in $\phi $. Furthermore, each
variable $x_{i_{j}}$ occurs in one old clause from $\phi $ and in two new
clauses in $\phi ^{\prime }$.

We will use again the gadgets $\mcalg_{x_{i_{j}}}$ for each
variable $x_{i_{j}}$ with a small modification. If the variable $x_{i_{j}}$
occurs two times in $\phi ^{\prime }$ as $\overline{x_{i_{j}}}$, then the
vertices $\beta _{x_{i_{j}}}$ and $\gamma _{x_{i_{j}}}$ of $\mathcal{G}%
\xspace_{x_{i_{j}}}$ will correspond to the negative assignment of $x_{i_{j}}
$. Otherwise, if $x_{i_{j}}$ occurs two times in $\phi ^{\prime }$ as a
positive literal, then the vertices $\beta _{x_{i_{j}}}$ and $\gamma
_{x_{i_{j}}}$ of $\mcalg_{x_{i_{j}}}$ will correspond to the
positive assignment of $x_{i_{j}}$. Again we will create one vertex $v_{k}$
for every clause $C_{k}$ of $\phi ^{\prime }$. If the vertex $v_{k}$
corresponds to an old clause (i.e.~from the initial formula $\phi $) then we
connect it to the $\gamma $-vertices of the subgraphs $\mcalg%
_{x_{i_{j}}}$ that correspond to these literals. If $v_{k}$ corresponds to a
new clause in $\phi ^{\prime }$ then this clause is of the form $(\overline{%
x_{i_{j}}}\vee x_{i_{j+1}})$. In this case we connect the corresponding $v$%
-vertex with the vertex $\beta _{x_{i_{j}}}$, if the variable $x_{i_{j}}$
occurs two times as a negative literal in $\phi ^{\prime }$, or with the
vertex $\alpha _{x_{i_{j}}}$, if $x_{i_{j}}$ occurs two times as a positive
literal in $\phi ^{\prime }$. Similarly, the $v$-vertex is connected with
the vertex $\alpha _{x_{i_{j+1}}}$, if the variable $x_{i_{j+1}}$ occurs two
times as a negative literal in $\phi ^{\prime }$, or with $\beta
_{x_{i_{j+1}}}$, if $x_{i_{j+1}}$ occurs two times as a positive literal in $%
\phi ^{\prime }$. The weights of these edges will be the same as before,
i.e.~each edge between a $v$-vertex and a $\beta $-vertex has weight 1 and
between a $v$-vertex and an $\alpha $-vertex or a $\gamma $-vertex has
weight~3.

In order to prove that the constructed graph is bipartite, it is sufficient
to prove that there is no cycle with odd length. Let $A$ be the set of all $%
\alpha $-vertices and of all $\beta $-vertices and let $\Gamma $ be the set
of all $\gamma $-vertices. First note that any cycle in the graph $G$ must
contain at least two vertices from $A\cup \Gamma $. Furthermore note that,
by the above construction, every path that connects two different vertices
of the set $A$, without touching any vertex of the set $\Gamma $, has even
length. Similarly, every path that connects two different vertices of the
set $\Gamma $, without touching any vertex of the set $A$, has also even
length. Thus every cycle in $G$ that does not contain any vertex from $%
\Gamma $ (resp.~from $A$) has even length. Consider now a cycle in $G$ that
contains vertices from both sets $A$ and $\Gamma $. Then, if we traverse
this cycle in any direction, we will encounter the same number of transition
edges from set $A$ to set $\Gamma $ and from the set $\Gamma $ to the set $A$%
. Therefore the length of the cycle is even, and thus $G$ is bipartite.
Thus, using the same argumentation as in Lemmas \ref{lem:one} and \ref%
{lem:at-least}, we obtain the following theorem.

\begin{theorem}
\label{thm:bip-hard} \greedy is strongly NP-hard,
even on \emph{bipartite} graphs with maximum degree at most 3 and with at
most three different integer weight values.
\end{theorem}

\subsection{Hardness results for \greedyv and \textsc{%
GreedyEdge}\xspace problems\label{hardness-additional-subsec}}

Having established the hardness results for \greedy %
in Sections~\ref{APX-subsec} and~\ref{hardness-bipartite-subsec}, we now
prove that also the decision problems \greedyv and 
\greedye are also strongly NP-hard.

\begin{theorem}
\label{thm:twohard} The decision problems \greedyv and 
\greedye are strongly NP-hard, even on
graphs with at most five different edge weights.
\end{theorem}

\begin{proof}
For the proof we amend the construction of Section~\ref{construction-subsec}
and the proofs of Section~\ref{APX-subsec}. Instead of reducing from the MAX2SAT(3) problem, 
we provide a reduction from the decision problem 3SAT(3). In this problem we are given a formula $\phi$, 
in which every clause has at most 3 literals and every variable appears in at most 3 clauses, 
and the question is whether there exists a truth assignment that satisfies \emph{all clauses} of $\phi$. 
This problem is NP-hard~\cite{Ausiello1999}.

Let $\phi $ be a 3SAT(3) formula with $n$ variables and $m$ clauses. We construct from $\phi$ a weighted graph $G$ in the same way as in Section~\ref{construction-subsec}, 
with the only difference that now every $v$-vertex is connected to at most three (instead of at most two) distinguished vertices from the subgraphs $\mcalg_{x_i}$. 
By following exactly the same proofs of Lemmas~\ref{lem:one} and~\ref{lem:at-least}, we can prove that this graph $G$ has a greedy matching with weight at least $14n+k$ if and only if 
there exists a truth assignment that satisfies at least $k$ clauses of the 3SAT(3) formula $\phi$. 
Now we augment this graph $G$ to a new graph $G^{\prime }$ by adding two new vertices $u$ and $u^{\ast }$.
Vertex $u$ is adjacent in $G^{\prime }$ to all the $v$-vertices of $G$ with
edges of weight $\frac{1}{2}$, while vertex $u^{\ast }$ is adjacent in $%
G^{\prime }$ only to vertex $u$ with an edge of weight $\frac{1}{4}$. Note
that $G^{\prime }$ has five different edge weights.

Let $\mathcal{M}^{\prime }$ be a greedy matching in $G^{\prime }$ and let $%
\mathcal{M}$ be the restriction of $\mathcal{M}^{\prime }$ on the edges of
the graph $G$. Since every edge of $G$ has larger weight than every edge
that is adjacent to $u$ or to $u^*$ in $G^{\prime }$, it follows that $%
\mathcal{M}$ is also a greedy matching of $G$. Assume that each of the $m$ $v
$-vertices is matched in $\mathcal{M}$ with a vertex from the subgraphs $%
\mcalg_{x_i}$. Then clearly $(u,u^*) \in \mathcal{M}^{\prime }$%
. Conversely, assume that $(u,u^*) \in \mathcal{M}^{\prime }$. If there
exists at least one vertex $v_j$ that is not matched in $\mathcal{M}$ with
any vertex from a subgraph $\mcalg_{x_i}$, then the edge $(v_j,
u)$ with weight $\frac{1}{2}$ will be available to be matched in $\mathcal{M}%
^{\prime }$, and thus the edge $(u, u^*)$ with weight $\frac{1}{4}$ will not
belong to $\mathcal{M}^{\prime }$, a contradiction. Therefore, $(u,u^*) \in 
\mathcal{M}^{\prime }$ if and only if each of the $m$ $v$-vertices is
matched in $\mathcal{M}$ a vertex from the subgraphs $\mcalg%
_{x_i}$. Furthermore, it follows by the proofs of Lemmas~\ref{lem:one} and~%
\ref{lem:at-least} that each of the $m$ $v$-vertices is matched in $\mathcal{%
M}$ with a vertex from the subgraphs $\mcalg_{x_i}$ if and only if
the weight of the greedy matching $\mathcal{M}$ of $G$ is at least $14n+m$,
or equivalently, if and only if there exists a truth assignment that
satisfies all $m$ clauses of the formula $\phi$.

Summarizing, there exists a greedy matching $\mathcal{M}^{\prime }$ of the
graph $G^{\prime }$, in which the given edge $(u,u^*)$ (resp.~the given
vertex $u^*$) is matched, if and only if the formula $\phi$ is satisfiable. 
Thus, since 3SAT(3) is NP-hard, it follows that both decision problems \greedyv and \greedye 
are strongly NP-hard, even on graphs with at most five different edge weights.
\end{proof}

\section{Further natural parameters of \textsc{GreedyMatching}}

In this section we investigate the influence of two further natural parameters to the computational complexity of \greedy, 
other than the parameters maximum degree and number of different edge weights that we considered in Section~\ref{sec:hardness}. 
As the first parameter we consider in Section~\ref{ratio-subsec} the minimum ratio $\lambda_{0}$ 
between two consecutive weight values, and as the second parameter we consider in Section~\ref{sec:param2} 
the maximum cardinality $\mu$ of the connected components of $G(w_i)$, over all possible weight values $w_i$. 
We prove that \greedy has a \emph{sharp threshold} behavior with respect to each of these parameters $\lambda_{0}$ and $\mu$.

\subsection{Minimum ratio of consecutive weights\label{ratio-subsec}}

Here we consider the parameter $\lambda_0 = \min_{i} \lambda_i$, 
where $\lambda_i = \frac{w_i}{w_{i+1}}>1$ is the ratio between the $i$th pair of consecutive edge weights. 
First we prove that, if $\lambda_0 \geq 2$, then there exists at least one 
maximum weight matching of $G$ that is an optimum solution for \greedy on 
$G$, obtaining the next theorem.

\begin{theorem}
\label{thm:rg2} \greedy can be computed in polynomial time if $\lambda_0 \geq 2$.
\end{theorem}

\begin{proof}
Let \mcal be a maximum weight matching for $G$. Note that \mcal can be computed 
in polynomial time~\cite{Duan14}. Assume
that \mcal is not a greedy matching. We will construct from \mcal a greedy 
matching of $G$ which has the same weight as \mcal, as follows. Since \mcal is 
not greedy, there must exist at least one edge $e \notin \mcal$ and two
incident edges $e^{\prime }, e^{\prime \prime }\in \mcal$,
where each of the weights $w(e^{\prime })$ and $w(e^{\prime \prime })$ of
the edges $e^{\prime }$ and $e^{\prime \prime }$, respectively, is smaller
than the weight $w(e)$ of the edge $e$. Since $w(e^{\prime }),w(e^{\prime
\prime })<w(e)$ and $\lambda_0 \geq 2$, it follows that $w(e^{\prime
}),w(e^{\prime \prime }) \leq \frac{w(e)}{2}$, and thus $w(e) \geq
w(e^{\prime }) + w(e^{\prime \prime })$. On the other hand $w(e) \leq
w(e^{\prime }) + w(e^{\prime \prime })$, since \mcal is a
maximum weight matching by assumption. Therefore $w(e) = w(e^{\prime }) +
w(e^{\prime \prime })$, and thus we can replace in \mcal the
edges $e^{\prime },e^{\prime \prime }$ with the edge $e$ without changing
the weight of \mcal.

We call all such edges $e\notin \mcal$ ``problematic''. Among
all problematic edges pick one edge $e$ with the maximum weight and replace
its incident matched edges $e^{\prime },e^{\prime \prime }$ with the edge $e$
in \mcal. We repeat this procedure until no problematic edge
is left, and thus we obtain a greedy matching $\mathcal{M}^{\prime }$ with
equal weight as \mcal. 
At each iteration the choice of the maximum weight problematic edge ensures
that no new problematic edges are created. We perform at most $|\mathcal{M} %
\xspace|/2 = |E|/4$ iterations, and thus $\mathcal{M}^{\prime }$ is computed
in polynomial time.
\end{proof}

\medskip

Recall that in the proof of the Theorem~\ref{thm:inai} the weight values 1,
3, and 4 were used, thus the \greedy is hard for $%
\lambda_0 \leq 4/3$. In the next theorem we amplify this result by showing
that \greedy is NP-hard for any
constant $\lambda_0 <2$. That is, complexity of \greedy has a threshold behavior at the parameter value $%
\lambda_0 = 2$.

\begin{theorem}
\label{thm:parlb} \greedy is \emph{strongly NP-hard} and \emph{APX-complete} 
for any constant $\lambda_0 < 2$, even on graphs with maximum degree at most 3 and with at most three different integer weight values.
\end{theorem}

\begin{proof}
For the proof we amend the weight values in the construction of Section~\ref{construction-subsec} 
and the proofs of Section~\ref{APX-subsec}. 
More specifically, in the construction of the graph $G$ from the formula $\phi$ 
in Section~\ref{construction-subsec}, we replace each edge of weight 4 with
an edge of weight $2x$, and each edge of weight 3 with an edge of weight $x+1$, 
where $x>1$ is an arbitrary integer. In particular, the results of
Sections~\ref{construction-subsec} and~\ref{APX-subsec} are given for the value $x=2$. 
By the proofs of Lemmas~\ref{lem:one} and~\ref{lem:at-least} 
(adapted for these new weights) it follows that there exists
a truth assignment that satisfies at least $k$ clauses of the 2SAT(3)
formula $\phi$ if and only if there is a greedy matching with weight at least $(6x+2)n+k$ in the constructed graph $G$. 
Similarly to Sections~\ref{construction-subsec} and~\ref{APX-subsec}, 
this graph $G$ maximum maximum degree 3 and three different integer weight values. 
Furthermore, $\lambda _{0}=\frac{2x}{x+1}$ can go arbitrarily close to 2 as $x$ increases. 
The statement of the theorem follows exactly by the proof of Theorem~\ref{thm:inai}, adapted for these new weights of the edges of $G$.
\end{proof}

\subsection{Maximum edge cardinality of a connected component in $G(w_{i})$\label{sec:param2}}

Another parameter that we can consider is the maximum edge cardinality $\mu$ of
the connected components of $G(w_{i})$, among all different weights $w_{i}$. 
Since $\mu =1$ implies that there is a unique greedy matching for 
$G$ which can be clearly computed in polynomial time, we consider the case $\mu \geq 2$.
In the original construction of Section~\ref{construction-subsec}, in every
gadget $\mathcal{G}_{x_{i}}$ there is a path with five edges where each edge
has weight 4. Thus $\mu =5$ in the graph $G$ of Section~\ref{construction-subsec}.
To prove our hardness result for $\mu =2$ in Theorem \ref%
{thm-maximum-component-parameter}, we modify the gadgets $\mathcal{G}_{x_{i}}$ 
as illustrated in Figure \ref{gadget-cardinality-parameter-fig}.
\begin{figure}[h]
\label{fig:modgx}
\par
\begin{center}
\begin{tikzpicture}[scale=1, auto,swap]

\node[draw,circle,inner sep=2pt] (01) at (0,1) {$\beta_{x_i}$};
\node[draw,circle,inner sep=2pt,fill] (10) at (1,0) [label=below:$p_{x_i}$]{};
\path[edge] (01) -- node[weight] {$2$} (10);
\node[draw,circle,inner sep=2pt,fill] (21) at (2,1) [label=above:$q_{x_i}$]{};
\path[edge] (21) -- node[weight] {$4$} (10);
\node[draw,circle,inner sep=2pt,fill] (30) at (3,0) [label=below:$r_{x_i}$]{};
\path[edge] (21) -- node[weight] {$5$} (30);
\node[draw,circle,inner sep=2pt] (41) at (4,1) {$\alpha_{x_i}$};
\path[edge] (41) -- node[weight] {$5$} (30);
\node[draw,circle,inner sep=2pt] (51) at (6,1) {$\gamma_{x_i}$};
\path[edge] (41) -- node[weight] {$4$} (51);
\node[draw,circle,inner sep=2pt,fill] (70) at (7,0) [label=below:$y_{x_i}$]{};
\path[edge] (51) -- node[weight] {$5$} (70);
\node[draw,circle,inner sep=2pt,fill] (81) at (8,1) [label=above:$z_{x_i}$]{};
\path[edge] (81) -- node[weight] {$5$} (70);
\node[draw,circle,inner sep=2pt,fill] (90) at (9,0) [label=below:$s_{x_i}$]{};
\path[edge] (81) -- node[weight] {$4$} (90);
\node[draw,circle,inner sep=2pt,fill] (101) at (10,1) [label=above:$t_{x_i}$]{};
\path[edge] (101) -- node[weight] {$2$} (90);
\end{tikzpicture}
\end{center}
\par
\caption{The modified weights for the gadget $\mcalg_{x_{i}}$.}
\label{gadget-cardinality-parameter-fig}
\end{figure}
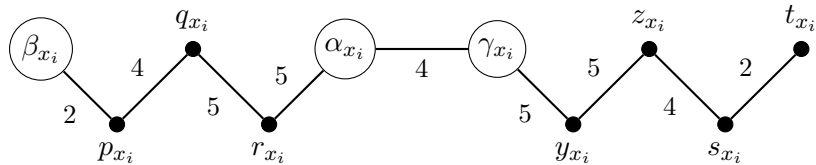

Notice that in every subgraph $\mathcal{G}_{x_{i}}$ (see Figure~\ref%
{gadget-cardinality-parameter-fig}) the connected components of each weight
have edge cardinality at most 2. Furthermore, the weight of the edge between
a $v$-vertex and a $\beta$-vertex has weight 1, while the edges between 
a $v$-vertex and an $\alpha$-vertex or a $\gamma$-vertex have weight~3. 
Thus these edges do not belong to any connected component with edges from 
$\mathcal{G}_{x_{i}}$. However, each $v$-vertex is connected with at most
two distinguished vertices in different gadgets $\mathcal{G}_{x_{i}}$ and $%
\mathcal{G}_{x_{i^{\prime }}}$. Therefore $\mu =2$ in the graph $G$ of this
modified construction. Considering these updated gadgets $\mathcal{G}_{x_{i}}
$ and using the same argumentation as in Lemmas~\ref{lem:one} and~\ref%
{lem:at-least}, we obtain that there is a greedy matching with weight at
least $18n+k$ in the constructed graph $G$ if and only if there is a truth
assignment that satisfies at least $k$ clauses from the original 2SAT(3)
formula $\phi $, which implies the next theorem.

\begin{theorem}
\label{thm-maximum-component-parameter} 
\greedy is \emph{strongly NP-hard} and \emph{APX-complete} for ${\mu \geq 2}$,
even on graphs with maximum degree at most 3 and with at most five different 
integer weight values.
\end{theorem}

\section{A randomized approximation algorithm\label{sec:approx}} 

In this section we provide a randomized approximation algorithm (\rgma) for \greedy with approximation ratio $\frac{2}{3}$ on two special classes of graphs (cf.~Section~\ref{two-bush-subsec}). 
Furthermore we highlight an unexpected relation between \rgma and the randomized \mrg algorithm for greedily approximating the maximum cardinality matching (cf.~Section~\ref{bush-cardinality}), 
the exact approximation ratio of which is a long-standing open problem~\cite{PS12,ADFS95,DF91}. 
Before we present our randomized algorithm \rgma, we first introduce the following class of weighted graphs, called \emph{bush graphs}.

\begin{definition}[Bush graph]
\label{def:bushg}
An \emph{edge-weighted} graph $G = (V, E)$ with $\ell$ edge weight values $w_1 > w_2 > \ldots > w_{\ell}$ 
is a \emph{bush graph} if, for every $i\in\{1,2,\ldots,\ell\}$, the edges of $G(w_{i})$ form a \emph{star}, 
which we call the \emph{$i$-th bush} of $G$.
\end{definition}

\begin{tcolorbox}[title=\rgma]
\textbf{Input:} Bush Graph $G$ with edge weight values $w_1 > w_2 > \ldots > w_{\ell}$.\\
\textbf{Output:} A greedy matching \mrgma.\vspace{-2mm}
\begin{enumerate}
\item $\mrgma \leftarrow \emptyset$
\vspace{-2mm}
\item \textbf{for} $i = 1 \ldots \ell$ \textbf{do}
\vspace{-2mm}
\item \hspace*{2mm} \textbf{if} $G_i \neq \emptyset$
\vspace{-2mm}
\item \hspace*{5mm} Select uniformly at random an edge $e_i \in G_i$ and add $e_i$ to \mrgma
\vspace{-2mm}
\item \hspace*{5mm} Remove from $G$ the endpoints of $e_i$ and all edges of $G_{i}$
\end{enumerate}
\end{tcolorbox}

\subsection{Bush graphs and the maximum cardinality matching\label{bush-cardinality}}

In this section we present the connection of the problem \greedy on (weighted) bush graphs 
to the problem of approximating the maximum cardinality matching in unweighted graphs via randomized greedy algorithms, cf.~Theorem~\ref{thm:cardappx}. 
Notice that we cannot directly apply the \rgma algorithm on unweighted graphs, 
since the algorithm has to consider the different bushes in a specific total 
order which is imposed by the order of the weights. 
Thus, in order to approximate a maximum cardinality matching in a given unweighted graph $G$ using the \rgma algorithm, 
we first appropriately convert $G$ to a (weighted) bush graph $G^*$ using the next Bush Decomposition algorithm, and then we apply \rgma on $G^*$.

\medskip

\begin{tcolorbox}[title= Bush Decomposition]
\textbf{Input:} Unweighted graph $G=(V,E)$ and $\eps \ll \frac{1}{|V|^3}$.\\
\textbf{Output:} A (weighted) bush graph $G^*$.\vspace{-2mm}
\begin{enumerate}
\item Set $k \leftarrow 0$
\vspace{-2mm}
\item \textbf{while} $E \neq \emptyset$ \textbf{do}
\vspace{-2mm}
\item \hspace*{2mm} Chose a random vertex $u \in V$ \label{gws3}
\vspace{-2mm}
\item \hspace*{2mm} For every $v' \in S := \{v' \in V: (u, v') \in E\}$ set $w(u, v') = 1 - k\cdot \eps$
\vspace{-2mm}
\item \hspace*{2mm} Remove the edges of $S$ from $E$
\vspace{-2mm}
\item \hspace*{2mm} $k \leftarrow k+1$
\end{enumerate}
\end{tcolorbox}

\medskip

Any unweighted graph $G = (V,E)$ can be considered as a weighted graph with edge weights $w(u, v) = 1$ for every edge $(u, v) \in E$, 
and thus in this case $\opt(G)$ coincides with the maximum cardinality matching in $G$.
In the next lemma we relate $\opt(G^*)$ with $\opt(G)$.

\begin{lemma}
\label{lem:cardopt}
$\opt(G) \geq \opt(G^*) \geq \opt(G) - \frac{1}{n}$.
\end{lemma}
\begin{proof}
Assume that $\opt(G^*) \geq \opt(G)$. Then, since the weight of each edge of $G$ is 1 and the weight of each 
edge of $G^*$ is by construction smaller than 1, it follows that $\opt(G^*)$ has strictly more edges than 
$\opt(G)$. This is a contradiction, since $\opt(G)$ is a maximum cardinality matching of $G$. 
Therefore $\opt(G) \geq \opt(G^*)$.

To prove that $\opt(G^*) \geq \opt(G) - \frac{1}{n}$, we construct 
from a maximum cardinality matching \mcal of $G$ 
a maximum weight greedy matching $\mcal^*$ for $G^*$ with the same cardinality as \mcal, 
i.e.~$|\mcal^*|=|\mcal|$, as follows. 
Starting from $\mcal$, we sequentially visit all centers $x_1, x_2, \ldots$ 
of the bushes in the weighted graph $G^*$, in a decreasing order of their edge weights.
Whenever a center $x_{i}$ of a bush in $G^*$ is unmatched in $\mcal$, 
then all its neighbors must be matched. If one of these neighbors of $x_i$ 
is matched in the current matching with an edge that is lighter than the edges of the bush of $x_i$, 
then we swap one of these edges with an edge incident to $x_{i}$. 
That is, the only case where $x_i$ stays unmatched is when all neighbors of $x_i$ are matched 
with edges of larger weight in the current matching. In this case there exists a maximum
cardinality matching for $G$ such that the vertex $x_i$ is unmatched.
At the end we obtain a matching $\mcal^*$ with the same cardinality
as the initial matching $\mcal$, but now $\mcal^*$ is a greedy matching for 
$G^*$. Thus, since $|\mcal|=|\mcal^*|$ and the weight of $\mcal^*$ is 
$w(\mcal^*)\leq OPT(G^*)\leq OPT(G)=|\mcal|$, it follows that $|\mcal|-w(\mcal^*)$
is less than or equal to the sum of the weight differences that have been 
introduced by ``Bush Decomposition'', i.e.~$|\mcal|-w(\mcal^*)\leq\frac{1}{n}$,
and thus $OPT(G)-OPT(G^*)\leq \frac{1}{n}$.
\end{proof}

\medskip

With Lemma~\ref{lem:cardopt} in hand the next theorem follows:
\begin{theorem}
\label{thm:cardappx}
Let $\rho$ be the approximation guarantee of \rgma algorithm on every bush graph. 
Then, for every $\eps<1$, \rgma computes a $(\rho-\eps)$-approximation of the 
maximum cardinality matching for unweighted graphs. 
\end{theorem}
We conjecture that a tight bound for $\rho$ is $\frac{2}{3}$; in Section~\ref{two-bush-subsec} we prove our conjecture true for two subclasses of bush graphs. 
Note that, although vertex $u$ in Step~\ref{gws3} of the Bush decomposition is selected at random, 
we do not use anywhere this fact in the proof of Lemma~\ref{lem:cardopt}. 
In particular, both Lemma~\ref{lem:cardopt} and Theorem~\ref{thm:cardappx} hold even when the choice of $u$ in Step~\ref{gws3} is \emph{arbitrary}.

\subsection{Randomized greedy matching in subclasses of bush graphs\label%
{two-bush-subsec}}

In this section we prove that our \textsc{Rgma}\xspace algorithm achieves an
approximation ratio $\rho=\frac{2}{3}$ in two special classes of bush
graphs, cf.~Theorems~\ref{thm:app2bush} and~\ref{thm:appbush}. Before we
prove these two theorems we first need to prove the following three lemmas
which will be useful for our analysis.

\begin{lemma}
\label{claim0-lem} Let $v$ be the center of the bush with the largest weight
in the graph $G+v$. If the edge $(u,v)$ belongs to a maximum greedy matching
of $G+v$ then $\opt(G+v)=w(u,v)+\opt(G-u)$.
\end{lemma}

\begin{proof}
Let $\mathcal{M}$ be a maximum greedy matching of $G+v$ that contains the
edge $(u,v)$. Then $\mathcal{M}\setminus \{(u,v)\}$ is a greedy matching of
the graph $G-u$, and thus its weight is at most $\opt(G-u)$.
That is, $\opt(G+v)-w(u,v)\leq \opt(G-u)$,
and thus $\opt(G+v)\leq w(u,v)+\opt(G-u)$.

Conversely, let now $\mathcal{M}^{\prime }$ be a maximum greedy matching of $%
G-u$. Since $v$ is the center of the bush with the largest weight in $G+v$,
it follows that every edge of $G-u$ has weight less than $w(u,v)$. Therefore 
$\mathcal{M}^{\prime }\cup \{(u,v)\}$ is a greedy matching of $G+v$, and
thus its weight is at most $\opt(G+v)$. That is, $\opt(G-u)+w(u,v)\leq \opt(G+v)$.
\end{proof}

\begin{lemma}
\label{claim1-1-lem}Let $w_{0}$ be the largest edge weight $G$ and let $u$
be a vertex $G$. Then $\opt(G-u)\geq \opt(G)-w_{0}$.
\end{lemma}

\begin{proof}
Let $\mathcal{M}$ be a maximum greedy matching of $G$. If $u$ is not matched
in $\mathcal{M}$ then $\opt(G-u)=\opt(G)$, which satisfies
the statement of the lemma. Suppose now that $u$ is matched in $\mathcal{M}$
and let $(u,v)\in \mathcal{M}$. We will modify the matching $\mathcal{M}$ of 
$G$ to a matching $\mathcal{M}^{\prime }$ of $G-u$ as follows. First remove
the edge $(u,v)$ from $\mathcal{M}$ and let $\mathcal{M}_{0}=\mathcal{M}%
\setminus \{(u,v)\}$. If $\mathcal{M}_{0}$ is a greedy matching of $G-u$
then define $\mathcal{M}^{\prime }=\mathcal{M}_{0}$; note that in this case $%
w(\mathcal{M}^{\prime })=w(\mathcal{M})-w(u,v)\geq w(\mathcal{M})-w_{0}$.
Otherwise, if $\mathcal{M}_{0}$ is not greedy, $v$ must have either (a) a
neighbor $v^{\prime }$ such that $v^{\prime }$ is unmatched in $\mathcal{M}%
_{0}$ or (b) a neighbor $v_{1}$ that is matched in $\mathcal{M}_{0}$ with an
edge $(v_{1},v_{2})$, where $w(v_{1},v_{2})<w(v,v_{1})$. If $v$ has both
such neighbors $v^{\prime }$ and $v_{1}$, we choose to consider only case
(a) if $w(v,v^{\prime })>w(v,v_{1})$, or only case (b) if $%
w(v,v_{1})>w(v,v^{\prime })$, breaking ties arbitrarily if $w(v,v^{\prime
})=w(v,v_{1})$. In case (a) we define $\mathcal{M}^{\prime }=\mathcal{M}%
_{0}\cup \{(v,v^{\prime })\}$; then $\mathcal{M}^{\prime }$ is greedy and 
\begin{eqnarray*}
w(\mathcal{M}^{\prime }) &=&w(\mathcal{M})-w(u,v)+w(v,v^{\prime }) \\
&\geq &w(\mathcal{M})-w(u,v) \\
&\geq &w(\mathcal{M})-w_{0}.
\end{eqnarray*}

In case (b) we define $\mathcal{M}^{\ast }=\mathcal{M}_{0}\cup
\{(v,v_{1})\}\setminus \{(v_{1},v_{2})\}$. In this case 
\begin{eqnarray*}
w(\mathcal{M}^{\ast }) &=&w(\mathcal{M})-w(u,v)+\left(
w(v,v_{1})-w(v_{1},v_{2})\right)  \\
&>&w(\mathcal{M})-w(u,v) \\
&\geq &w(\mathcal{M})-w_{0}.
\end{eqnarray*}%
If $\mathcal{M}^{\ast }$ is greedy then we define $\mathcal{M}^{\prime }=%
\mathcal{M}^{\ast }$. Otherwise, if $\mathcal{M}^{\ast }$ is not greedy, $%
v_{2}$ must have (similarly to the above) either (a) a neighbor $%
v_{2}^{\prime }$ such that $v_{2}^{\prime }$ is unmatched in $\mathcal{M}^{\ast }
$ or (b) a neighbor $v_{3}$ that is matched in $\mathcal{M}^{\ast }$ with an
edge $(v_{3},v_{4})$, where $w(v_{3},v_{4})<w(v_{2},v_{3})$. We continue to
update the matching $\mathcal{M}^{\ast }$ as above, until we reach a
matching $\mathcal{M}^{\prime }$ of $G-u$ such that $\mathcal{M}^{\prime }$
is greedy and $w(\mathcal{M}^{\prime })\geq w(\mathcal{M})-w_{0}$. 
This completes the proof of the lemma, since $w(\mathcal{M}%
^{\prime })\leq \opt(G-u)$ and $w(\mathcal{M})=\opt(G)$.
\end{proof}

\medskip

In the next theorem we prove that \rgma achieves an approximation ratio of $\frac{2}{3}$ when 
applied to a bush graph with only two different edge weights. 
Using this theorem as the induction basis, we then prove in Theorem~\ref{thm:app2bush} that 
\rgma achieves an approximation ratio of $\frac{2}{3}$ also when applied to a bush graph in which every bush has at most two edges.

\begin{theorem}
\label{thm:app2bush}
\rgma is a $\frac{2}{3}$-approximation when applied on bush graphs with only two
weights and with an arbitrary number of edges per bush.
\end{theorem}

\begin{proof}
Let $g_1$ and $g_2$ be the two bushes and let $x_2$ be the center of the bush 
$g_2$. We have to consider the following three cases: $x_2$ is the center of the 
bush $g_1$ too, $x_2$ does not belong to the bush $g_1$, $x_2$ is a leaf of the 
bush $g_1$. 
If $x_2$ is the center of the bush $g_1$ too, then any greedy matching is 
consisted by only one edge from $g_1$, thus \rgma always finds an optimal greedy
matching. 
When $x_2$ does not belong to the bush $g_1$ we can partition the edges of $g_2$
in three sets: 
\begin{itemize}
\item $g_{21}$: edges of $g_2$ that are incident to the center of the bush $g_1$.
\item $g_{22}$: edges of $g_2$ that are not incident to leaves of $g_1$.
\item $g_{23}$: edges of $g_2$ that are incident to a leaf of $g_1$.
\end{itemize}
Without loss of generality we may assume that $g_{21} = \emptyset$, since none 
of these edges can be chosen in any greedy matching. Furthermore, if 
$g_{22} \neq \emptyset$, then \rgma will always choose one edge from the $g_2$,
thus it will return an optimal greedy matching. Finally, if $|g_{23}|>1$, then 
for every choice of edge from $g_1$, there will be at least one edge in $g_2$ 
that can be chosen in the greedy matching. Hence \rgma will construct an optimal
greedy matching. Finally suppose that $e' \in g_{23}$ be the unique edge $g_{23}$
has. Then the probability that $e'$ is deleted in the first iteration of \rgma
is $\frac{1}{|g_1|}$, thus \rgma will return a matching with expected weight 
$w_1 + \frac{|g_1| - 1}{|g_1|}w_2 \geq w_1 + \frac{1}{2}w_2$. Hence the 
approximation guarantee of \rgma in this case is 
$\frac{w_1 + w_2/2}{w_1 + w_2} > \frac{3}{4}$. 

Finally, we consider the case where $x_2$ is a leaf of the bush $g_1$. Then we 
may partition the edges of $g_2$ into the following two sets:
\begin{itemize}
\item $g_{24}$: edges of $g_2$ that are adjacent only to one leaf of $g_1$.
\item $g_{25}$: edges of $g_2$ that are adjacent to two leaves of $g_1$.
\end{itemize}
If $g_{24} \neq \emptyset$, then with probability $\frac{|g_1|-1}{|g_1|}$ the 
algorithm returns a matching with weight $w_1 + w_2$. Thus, if this is the case
the approximation guarantee of the algorithm is at least $\frac{3}{4}$. Hence,
without loss of generality we may assume that $g_{24} = \emptyset$.  If 
$|g_{25}| > 1$, then the only case the algorithm \rgma returns a matching with
weight $w_1$ is when the edge $(x_1, x_2)$ is chosen; in every other case there
exists at least one in $g_2$ that it is not deleted. This happens with 
probability $\frac{1}{|g_1|}$. Notice that by our assumption it must be true 
that $|g_1| \geq |g_2| > 1$. Thus, the expected weight of the matching returned 
by the algorithm is $\frac{|g_1|-1}{|g_1|}(w_1 + w_2) + w_1/|g_1| \geq 
\frac{1}{2}(w_1 + w_2) + w_1/2$. Hence the approximation guarantee of the \rgma 
is better than $\frac{2}{3}$. The remaining case is when there is a unique edge 
$(x_2, v) \in g_{25}$. Then the \rgma returns a matching of weight $w_1$ when at 
least one of the edges $x_2$ and $v$ is chosen in the first iteration. This 
happens with probability $\frac{2}{|g_1|}$. Hence the expected weight of the 
matching returned by the algorithm is $\frac{|g_1|-2}{|g_1|}(w_1 + w_2) + 
2w_1/|g_1| \geq \frac{1}{3}(w_1 + w_2) + w_1$. Hence, in this case the 
approximation guarantee of the algorithm is at least $\frac{2}{3}$.
\end{proof}

\medskip

We are now ready to prove the main theorem of this section.
\begin{theorem}
\label{thm:appbush} \textsc{Rgma}\xspace is a $\frac{2}{3}$-approximation
when applied on bush graphs where each bush has at most two edges.
\end{theorem}

\begin{proof}
We will prove the claim by induction on the number of the different weight
values $w_{1},\ldots ,w_{\lambda }$ the bush graph $G$ has. We will use $%
\av\xspace(G)$ to denote the expected weight of the greedy matching
produced by \textsc{Rgma}\xspace on the bush graph $G$. We know from Theorem~%
\ref{thm:app2bush} that the claim holds when there are only two weight
values in $G$. Assume that for any bush graph $G$ with $i\geq 2$ different
weight values such that every bush of $G$ has at most two edges, it holds
that $\av\xspace(G)\geq \frac{2}{3}\opt(G)$. We
will prove that the claim holds also for bush graphs with $i+1$ different
weight values.

Let $x_{0}$ be the center of the bush with the largest weight $w_{0}$ and
let $\alpha $ and $\beta $ be the leaves of this bush. Without loss of
generality we can assume that all incident edges to $x_{0}$ have weight $%
w_{0}$, since every other incident edge of $x_{0}$ with weight $w_{i}<w_{0}$
would never be selected by the greedy algorithm. Assume that the edge $%
(x_{0},\alpha )$ belongs to the optimal greedy matching of $G+x_{0}$. Thus, $%
\opt(G+x_{0})=w_{0}+\opt(G-\alpha )$ by
Lemma \ref{claim0-lem}. Furthermore $\opt(G-\beta )\geq 
\opt(G)-w_{0}$ by Lemma \ref{claim1-1-lem}. Hence we get 
\begin{align*}
\av\xspace(G+x_{0})& =w_{0}+\frac{1}{2}(\av\xspace%
(G-\alpha )+\av\xspace(G-\beta )) \\
& \geq w_{0}+\frac{1}{2}\cdot \frac{2}{3}(\opt(G-\alpha )+%
\opt(G-\beta )) \\
& =\frac{1}{3}(\opt(G+x_{0})+\opt(G-\beta
)+2w_{0}).
\end{align*}%
Now note that $\opt(G-\beta )\geq \opt(G)-w_{0}$
and $\opt(G)\geq \opt(G+x_{0})-w_{0}$ by Lemma~\ref{claim1-1-lem}. 
Therefore $\opt(G-\beta)+2w_{0}\geq \opt(G)+w_{0}\geq \opt(G+x_{0})$, and thus
\begin{equation*}
\av\xspace(G+x_{0})\geq \frac{1}{3}(\opt(G+x_{0})+%
\opt(G+x_{0}))=\frac{2}{3}\opt(G+x_{0}).
\end{equation*}
\end{proof}

\section{Conclusions}

Several interesting open questions stem from our paper. 
Probably the most important one is to derive tight approximation guarantees 
$\rho$ for the maximum weight greedy matching problem, even for bush graphs. 
We conjecture that $\rho=\frac{2}{3}$; an affirmative answer to our conjecture 
would imply that the algorithm \mrg for maximum cardinality matching in 
unweighted graphs has an approximation ratio of almost $\frac{2}{3}$, thus 
solving a longstanding open problem~\cite{PS12,ADFS95,DF91}. We believe that our
approach might provide novel ways of better analysis of the \mrg algorithm.
As we proved, \greedy is NP-hard even on graphs of maximum degree three with at 
most three different weight values on their edges. It remains open whether \greedy can be solved 
efficiently when there are only two weight values on the edges of the input graph.

\end{document}